\documentclass[10pt,conference]{IEEEtran}

\usepackage{amsmath,amssymb,amsfonts, amsthm}

\usepackage{algorithmicx}
\usepackage{algorithm}
\usepackage{algpseudocode}

\usepackage{textcomp}
\usepackage{xcolor}
\usepackage{caption}
\usepackage{url}
\usepackage{xurl}
\usepackage[numbers]{natbib}
\usepackage{tabularx}
\usepackage{adjustbox}

\newtheorem{theorem}{Theorem}
\newtheorem{lemma}{Lemma}
\newtheorem{proposition}{Proposition}
\newtheorem{corollary}{Corollary}
\newtheorem{definition}{Definition}

\def\BibTeX{{\rm B\kern-.05em{\sc i\kern-.025em b}\kern-.08em
    T\kern-.1667em\lower.7ex\hbox{E}\kern-.125emX}}
\begin{document}

\title{MIST: An Efficient Approach for Software-Defined Multicast in Wireless Mesh Networks}

\author{
    \IEEEauthorblockN{Rupei Xu\IEEEauthorrefmark{1}, Yuming Jiang\IEEEauthorrefmark{2}, Jason P.~Jue\IEEEauthorrefmark{1}}
    \IEEEauthorblockA{\IEEEauthorrefmark{1}The University of Texas at Dallas, UAS, Email: \{rupei.xu, jjue\}@utdallas.edu}
    \IEEEauthorblockA{\IEEEauthorrefmark{2}Norwegian University of Science and Technology, Norway, Email: \{yuming.jiang\}@ntnu.no}
}

\maketitle

\begin{abstract}
Multicasting is a vital information dissemination technique in Software-Defined Networking (SDN). With SDN, a multicast service can incorporate network functions implemented at different nodes, which is referred to as software-defined multicast. Emerging ubiquitous wireless networks for 5G and Beyond (B5G) inherently support multicast. However, the broadcast nature of wireless channels, especially in dense deployments, leads to neighborhood interference as a primary system degradation factor, which introduces a new challenge for software-defined multicast in wireless mesh networks. To tackle this, this paper introduces an innovative approach, based on the idea of minimizing both the total length cost of the multicast tree and the interference at the same time.  Accordingly, a novel bicriteria optimization problem is formulated--\emph{Minimum Interference Steiner Tree (MIST)}, which is the edge-weighted variant of the vertex-weighted secluded Steiner tree problem \cite{chechik2013secluded}. To solve the bicriteria problem, instead of resorting to heuristics, this paper employs an innovative approach that is an approximate algorithm for MIST but with guaranteed performance. Specifically, the approach exploits the monotone submodularity property of the interference metric and identifies Pareto optimal solutions for MIST, then converts the problem into the submodular minimization under Steiner tree constraints, and designs a two-stage relaxation algorithm. Simulation results demonstrate and validate the performance of the proposed algorithm. 
\end{abstract}


\section{Introduction}
\label{intro}

The evolution of 5G and beyond (B5G) networks has ushered in a diverse range of emerging wireless communication services that are well-suited to support pervasive and ambient application scenarios. Within this dynamic landscape, Software-Defined Networking (SDN) and Network Function Virtualization (NFV) emerge as pivotal enablers. Their integration has the potential to substantially enhance wireless communications, providing advanced service capabilities and increased flexibility. The convergence of SDN and NFV technologies with wireless networks holds tremendous promise for reshaping the future of communication networks. It represents a powerful synergy that can unlock new opportunities and drive innovation in the evolving field of wireless communication. 

SDN has brought significant advancements to various network services, including multicast. Multicast involves the transmission of data from one sender to multiple receivers in the network. The programmability, centralized control, and flexibility of SDN have led to the development of various multicast-based applications such as Content Distribution Networks (CDNs), mobile social networks, Internet Protocol Television (IPTV), etc. As in NFV, with SDN, a multicast service can incorporate network functions (NFs) implemented at different nodes, providing software-defined multicast. 

With the rise of ubiquitous wireless networks for B5G, the inherent broadcast nature of wireless communications, particularly in densely deployed areas, contributes significantly to neighborhood interference, serving as a primary factor leading to network performance degradation. For software-defined multicast in such networks, it is hence imperative to take interface into consideration. 

The objective of this paper is to find an interference-aware multicast tree that can be used to support software-defined multicast in wireless mesh networks (WMN).  Unlike the traditional multicast problem, which has been extensively investigated, the problem of software-defined multicast in WMN has two unique characteristics. One is the interference challenge. Another is that the multicast service involves network functions that are implemented (not at the sender but)  at different nodes in the network. For instance, a wireless network has multiple tenants and a tenant has the need to communicate with the nodes where functions for its service are implemented. A more specific example is a wireless sensor network. A tenant or user may need to send and collect information from multiple sensor nodes in the network.

In the literature, SDN multicasting has been extensively investigated \cite{alsaeed2018multicasting}, though the surveyed studies typically assume no interference among vertices in the network. Recently, \cite{mirjalily2021interference} has filled this gap by opening an interesting new research direction: interference-aware NFV-enabled multicast in WMN. However, its focus is on how to optimally place the virtualized network functions (VNFs) in the network such that the considered total multicast cost in terms of installation, computing, storage, and resource consumption is minimized. Despite the effort of the numerical approximate heuristic approach proposed in \cite{mirjalily2021interference} to solve the optimization problem, the theoretical performance of the approximation remains unexplored.  

In this paper, a novel combinatorial approach is introduced to tackle the problem of interference-aware software-defined multicast in WMN. 
Specifically, it is formulated as a bicriteria optimization problem, called \emph{Minimum Interference Steiner Tree (MIST)}. In MIST, a software-defined multicast service with functions implemented at different nodes is mapped to a Steiner tree where the nodes implementing those NFs are terminals of the Steiner tree. In addition, this multicast Steiner tree is minimized by minimizing the interference at the same time. The approach exploits the monotone submodularity property of the interference metric and identifies Pareto optimal solutions for MIST, then converts the problem into the submodular minimization under Steiner tree constraints, and designs a two-stage relaxation algorithm.

The key contributions of this paper can be summarized as follows:

\begin{itemize}
	\item \emph{Incorporation of interference metric into software-defined multicast in wireless mesh networks.}
	\item \emph{A novel framework for bicriteria optimization.}
	\item \emph{Monotone submodularity proof of the interference metric and Pareto optimal solution characteristics using vicinal preorder.}
	\item \emph{Theoretically proven performance-guaranteed algorithm design.}
\end{itemize}

The remainder of this paper is structured as follows. Section \ref{related} provides a review of related works in the field. Following that, in Section \ref{sys-formula},  the system model and the problem formulation as a bicriteria optimization problem are introduced. The proposed algorithmic approaches are discussed in Section \ref{algo}. Simulation results are presented in Section \ref{simulation}, and finally, Section \ref{con} concludes the paper.

\section{Related Works}
\label{related}

In the realm of Software-Defined Networking (SDN) multicasting, several survey papers have made significant contributions to the understanding of this technology. \cite{nunes2014survey} provides a comprehensive overview of SDN and its applications, including multicast, offering valuable insights into the historical context and future directions of SDN technology. Furthermore, \cite{gu2015survey} and \cite{alsaeed2018multicasting} conduct in-depth surveys specifically focused on SDN multicasting, providing comprehensive insights into the state of the field and its advancements.

The interference metric introduced in this paper can be likened to the \emph{Minimum $k$-Union Problem} \cite{chlamtavc2017minimizing} or the \emph{Minimum Neighborhood Problem} \cite{kabra2018parameterized}. These problems, while intuitive and natural, have indeed received relatively limited attention in the existing literature.

The \emph{Secluded Steiner Tree} problem \cite{chechik2013secluded}\cite{fomin2017parameterized} is indeed closely related to the problem addressed in this paper. The secluded Steiner tree connects a given set of terminals in a graph, such that the exposure (number of closed neighborhood) of the tree is minimized. However, in this paper, we consider the edge weight cost of the multicast tree, instead of the vertex cost in for the weighted version in \cite{chechik2013secluded}\cite{fomin2017parameterized}. Moreover, the unweighted secluded Steiner tree in \cite{chechik2013secluded} can be considered a minimum submodular Steiner tree as described in \cite{goel2009approximability}, since the closed neighborhood is proven to be monotone submodular in this paper.

As of now, \cite{mirjalily2021interference} stands among the few papers that have successfully integrated interference metrics into the multicast tree within SDN-enabled wireless mesh networks, based on integer linear programming approaches. Our paper adopts the same interference metric. However, it's important to note that our paper does not delve into multicast link scheduling issues, as explored in \cite{mirjalily2021interference}, to create a more versatile and comprehensive application. Since more advanced techniques are capable of managing simultaneous transmission collisions, which could further enhance the performance and applicability of multicast communication in complex network scenarios. What makes our paper particularly noteworthy is that, for the first time, it offers a more detailed exploration of the underlying combinatorial properties and presents a more rigorous algorithm with performance guarantees. 

\section{System Model and Problem Formulation}
\label{sys-formula}

\subsection{System Model}
\label{sys}
We consider a wireless mesh network of nodes or vertices, connected by links or edges. Among the nodes, there is a subset of nodes, referred to as terminals. In these terminal nodes, some network functions needed by multicast services are implemented. To ease representation and simplify the problem, we assume each terminal node performs only one unique function. For a multicast service, we assume that it involves such functions implemented at different terminal nodes. 

Before presenting the problem formulation, let's introduce two fundamental graph theory concepts using appropriate terminology. Given an undirected graph $G=(V, E)$ with non-negative edge weights $L$ and a subset of vertices $S\subseteq V,$ referred to as terminals, \emph{a Steiner tree} is a tree in $G$ that spans the set of terminals $S.$ In the graph $G=(V,E),$ for each vertex $v$ in the set $V,$ we define the set $N(v)\cup {v}$ as $N[v]$ representing the set of \emph{closed neighbors} of the vertex $v.$ In $G=(V,E),$ for $V'\subseteq V,$ we define $N[V']=\bigcup_{v\in V'} N[v],$ representing the set of \emph{closed neighbors} of the vertex set $V'\subseteq V.$

The wireless mesh network with a multicast request is modeled as $G=(r, V, E, F, L)$, where $V: |V|=n$ denotes the set of $n$ nodes or vertices, $E: |E|=m$, denotes the set of $m$ links or edges, $F_{v\in V}(v)\in [F_1, F_2, ..., F_n]$ denotes the network function on vertex $v$, $L: E\rightarrow \mathbb{R}^+$ denotes the lengths on edges, and $r$ is the sender or root of the multicast service. The set of receivers, or target terminals, of the multicast service is denoted as $S.$ A multicast request may involve multiple network functions and is denoted as $R=[\mathcal{F}_1, \mathcal{F}_2,...,\mathcal{F}_k]$. We use $V_R$ to denote the vertex set containing the virtual functions mapped from $R,$ and $v_{\mathcal{F}_j}$ to denote the vertex containing the virtual function $\mathcal{F}_j.$

\subsection{Problem Formulation}
\label{formula}

The primary objective is to construct a multicast tree that facilitates the establishment of connections between the source node $r$ and the target terminals $V_T,$ e.g. via SDN \cite{alsaeed2018multicasting}. 
However, unlike the traditional multicast problem, the multicast problem considered in this paper has two unique characteristics. One is that, the multicast tree must contain all nodes or vertices where the corresponding network functions are implemented. Accordingly, the multicast tree is a Steiner tree. In addition, due to the wireless nature, there is interference among adjacent links. Together, in constructing the multicast tree, we aim to simultaneously optimize two factors and will call the resulting problem the \emph{Minimum Interference Steiner Tree (MIST)} problem. The \emph{Minimum Interference Steiner Tree} problem in a wireless mesh network aims to establish a multicast tree (Steiner tree) efficiently connecting the source and target terminals while minimizing both the total communication length cost and the number of nodes experiencing interference in the tree's vicinity. 

\begin{algorithm}[H]
	\caption*{\textbf{Minimum Interference Steiner Tree (MIST)}}
		\renewcommand{\algorithmicrequire}{\textbf{INSTANCE:~}}
		\renewcommand{\algorithmicensure}{\textbf{SOLUTION:~}}
        \newcommand{\algorithmiccost}{\textbf{COST FUNCTION:~}}
        \newcommand{\algorithmicobjective}{\textbf{OBJECTIVE:~}}
		\algorithmicrequire 
		Wireless Mesh Network $G=(r, V, E, F, L)$ consisting of a root $r;$ a set of vertices $V: |V|=n;$ a set of edges $E:|E|=m;$ a unique virtual function $F_{v\in V}(v)\in [F_1, F_2, ..., F_n]$ on each vertex; latencies $L: E\rightarrow \mathbb{R}^+$ on edges. Multicast request $R=[\mathcal{F}_1, \mathcal{F}_2,...,\mathcal{F}_k].$
        \newline
        \newline
		\algorithmicensure
		A Steiner tree $T=(V_T, E_T)$ rooted at $r$, such that $r\cup V_R=S\subseteq V_T\subseteq V.$
		\newline
        \newline
        \algorithmiccost $(\sum_{e_i\in E_T} L(e_i),|N[V_T]|)$
		\newline
        \newline
        \algorithmicobjective Minimize, Minimize
\end{algorithm}

The first objective function bears a close relationship to the extensively studied Minimum Steiner Tree Problem, a classic NP-complete problem as identified by Karp \cite{Karp1972}. The subsequent two NP-complete problems \emph{Minimum $k$-Union Problem} and \emph{Minimum Neighborhood Problem} introduced in \cite{chlamtavc2017minimizing} and \cite{kabra2018parameterized} are intricately linked to the interference metric of the second objective function presented in this paper. 

\section{Algorithm Design}
\label{algo}

To solve the bicriteria  MIST problem, this section introduces the design of an algorithm. To this aim, we first explore the monotone submodularity properties of the interference metric and vicinal preorder properties of Pareto optimal solutions for MIST, based on which a two-stage relaxation algorithm with guaranteed performance is proposed. 

\subsection{Monotone Submodularity of the Interference Metric}

Let $\mathcal{U}$ be a universe set of $n$ elements. A function $f: 2^\mathcal{U} \rightarrow \mathbb{R},$ $f$ is monotone, if for any $\mathcal{C} \subseteq \mathcal{U}$ and $j \in \mathcal{U}$, $f(\mathcal{C} \cup {j}) \geq f(\mathcal{C}).$ This property ensures that adding an element to a subset can only increase the function value. Additionally, $f$ is defined as submodular if it satisfies the diminishing marginal returns property: for every $\mathcal{A} \subseteq \mathcal{B} \subseteq \mathcal{U}$ and $j \notin \mathcal{B},$ $f(\mathcal{A} \cup {j})-f(\mathcal{A}) \geq f(\mathcal{B}\cup {j})-f(\mathcal{B}).$ This property indicates that the marginal gain of adding an element to a smaller set is greater than or equal to adding it to a larger set. 

Given a graph $G=(V, E),$ for $v\in V,$ let $N(v)\cup {v}=N[v]$ be the set of closed neighbors of $v.$ The set coverage function, in general, does not exhibit a monotone behavior, as the addition of elements to a set may or may not lead to an increase in the function value. Moreover, it is not inherently submodular. However, it's worth noting that there are special cases where certain functions, such as the graph closed neighborhood function for the considered interference metric in the MIST problem, demonstrate monotonicity and submodularity properties as shown in Lemma \ref{sub}.

\begin{lemma}
	\label{sub}
	Given graph $G(V,E),$  $|N[V'\subseteq V]|=|\bigcup_{v\in V'}N[v]|$ is monotone submodular. 
\end{lemma}

\begin{proof}
Firstly, the monotonicity holds because the closed neighborhood only expands when more nodes are added. For any $A\subseteq B\subseteq V,$ the closed neighborhood of $A$ is a subset of the closed neighborhood of B, thus $|N[A]|\leq |N[B]|.$ 

Secondly, the submodularity property holds as well, for any $A\subseteq B\subseteq V,$ and let $v\notin B,$ adding $v$ to $A$ will result in a larger marginal gain compared to adding $v$ to $B.$ Since $|N[A]|\leq |N[B]|$ and $|N[A]\cap N[v]|\leq |N[B]\cap N[v]|,$ in addition, $N[A\cup \{v\}]=N[A]+N[v]-N[A]\cap N[v]$ and $N[B\cup \{v\}]=N[B]+N[v]-N[A]\cap N[v],$ therefore, for any $v\notin B,$ $|N[A\cup \{v\}]|-|N[A]|\geq |N[B\cup \{v\}]|-|N[B]|.$

Therefore, the graph closed neighborhood function satisfies both monotonicity and submodularity properties.
\end{proof}

In typical graph optimization problems, distinctions between edge-weighted and node-weighted scenarios often lead to different performance outcomes. However, with the closed neighborhood function, it can obtain identical sets of vertices within both the edge and vertex domains.

For each edge, its closed neighborhoods are the union of the closed neighborhoods of its two end vertices: $N[e]=\cup_{v\in V(e)} N[v].$ Note that in graph $G=(V, E),$ for any tree $T=(V_T, E_T)$ which spans the terminal set $S\subseteq V,$ we can get $N[V_T]= \cup_{v\in V_T} N[v] = \cup_{v\in V(E_T)} N[v] =\cup_{e\in E_T} N[e]=N[E_T].$ Next, we will present a rigorous proof demonstrating that the closed neighborhood function on the edge domain of the graph exhibits monotone submodularity, analogous to the proof provided for the vertex domain of the graph in Lemma \ref{sub}.

\begin{lemma}
\label{mist-edge-sub}

Given graph $G(V,E),$ for any tree $T=(V_T, E_T)$ which spans the terminal set $S\subseteq V,$ $|N[E_T]|=|\bigcup_{e\in E_T}N[e]|$ is monotone submodular. 

\end{lemma}

\begin{proof}
For any $E_1\subseteq E_2\subseteq E_T,$ we have $|N[E_1]|\leq |N[E_2]|,$ thus the monotonicity property holds. 

For any $E_1\subseteq E_2\subseteq E_T,$ let $e'\in E_T-E_2.$  When $e' \in E_1,$ the marginal gain $|N[E_1\cup e']|-|N[E_1]|$ is zero. When $e' \notin E_1,$ since $E_1 \subseteq E_2$, the set $\left(\bigcup_{e\in E_1} N[e]\right)$ is a subset of $\left(\bigcup_{e\in E_2} N[e]\right)$. Therefore, adding edge $e'$ to $E_2$ will not introduce any new closed neighborhoods from the edges in $E_1$, meaning that the marginal gain from adding $e'$ to $E_2$ is at most the marginal gain from adding $e'$ to $E_1$. Thus, we have shown that the number of closed neighborhoods of edges in a graph is monotone submodular.
\end{proof}

\subsection{Vicinal Preorder Properties of the Pareto Optimal Solutions}

The \emph{vicinal preorder} \cite{foldes1978dilworth} on $V(G)$ in graph $G(V, E)$ is defined as follows: $u \lesssim v$ if and only if $N(u)\subseteq N(v)\cup {v},$ where $N(v)$ is the set of neighbors of $v.$ $N(v)\cup {v}$ is the set of closed neighbors of $v,$ which can also be represented as $N[v].$ If neither $u \lesssim v$ nor $v \lesssim u,$ we refer to $u$ and $v$ as non-comparable, which is represented as $u \shortparallel v.$

\begin{proposition}
	\label{vicinal}
	Assume that tree $T^*=(V_T^*, E_T^*)$ is the Pareto optimal solution to the Minimum Interference Steiner Tree problem, if $v\notin N[V_T^*],$ and $\forall v'\in V_T^*,$ we can obtain either $v \lesssim v'$ or $v \shortparallel v'.$ 
\end{proposition}

\begin{proof}
	$\forall v'\in V_T^*,$ and $v\notin N[V_T^*],$ if $v' \lesssim v,$ according to the definition ofvicinal preorder, one can obtain $N(v')\subseteq N(v)\cup {v}.$ Therefore, if $v\in N(v'),$ then $N(v')-{v}\subseteq N(v),$ it contradicts $v\notin N[V_T^*];$ if $v\notin N(v'),$ then $N(v')\subseteq N(v),$ which also contradicts $v\notin N[V_T^*].$
\end{proof}

\begin{proposition}
	\label{leaf}
	Assume that tree $T^*=(V_T^*, E_T^*)$ is the Pareto optimal solution to the Minimum Interference Steiner Tree problem. In $V_T^*-r,$ the vertices of all the leaves are in the terminal set. 
\end{proposition}

\begin{proof}
	Assume $v'\in V_T^*-r$ and it is also the leaf vertex of $T^*.$ If it is not in terminal sets, removing it can reduce the total latency of $T^*;$ moreover, by removing it, either $|N[V_T^*]|$ decreases if $N(v')\cap N[V_T^*-r-{v'}]\neq\emptyset$, or $|N[V_T^*]|$ remains the same if $N(v')\cap N[V_T^*-r-{v'}]=\emptyset$. As a result, it contradicts the fact that $T^*$ is the Pareto optimal solution.
\end{proof}

\subsection{Monotone Submodular Minimization under Combinatorial Constraints}

Submodular Function Minimization (SFM) has been an important problem in combinatorial optimization. The unconstrained problem of SFM is known to be computable in polynomial time \cite{grotschel1981ellipsoid} \cite{cunningham1985submodular}. However, SFM with constraints poses a formidable challenge in the field of combinatorial optimization. 

To better address the unique characteristics of our problem, this paper introduces the novel reformulation \emph{Closed Neighborhood Minimum Steiner Tree}, defined as follows.

\begin{definition}\textbf{(Closed Neighborhood Minimum Steiner Tree )} Given a connected undirected graph $G = (V, E)$ with a terminal set $S \subseteq V$, and a Steiner Tree (ST) set $\mathcal{ST}$ on $G,$ the closed neighborhood function $N: 2^{V} \rightarrow \mathbb{R}^+,$ find $A \in \mathcal{ST}$ that minimizes $N[V(A)].$
\end{definition}

Next, we will proceed with the relaxation approach following the formulations of two single-objective problems, namely \emph{Closed Neighborhood Minimum Spanning Tree} and \emph{Closed Neighborhood $s{\text-}t$ Shortest Path}. These two problems can be seen as \emph{(Monotone) Submodular Minimum Spanning Tree} and \emph{(Monotone) Submodular $s{\text-}t$ Shortest Path}, as outlined in \cite{goel2009approximability}.

\begin{definition} \textbf{Closed Neighborhood Minimum Spanning Tree)} Given a connected undirected graph $G = (V, E),$ a Spanning Tree (SPT) set $\mathcal{SPT}$ on $G,$ find $A \in \mathcal{SPT}$ that minimizes $N[E(A)].$
\end{definition}

\begin{definition} \textbf{(Closed Neighborhood $s{\text-}t$ Shortest Path)}
Given a connected undirected graph $G = (V, E),$ a pair of vertices $s,t\in V,$ and a set of $s{\text-}t$ paths $st{\text-}\mathcal{P}$ on $G,$ find $A \in st{\text-}\mathcal{P}$ that minimizes $N[E(A)].$
\end{definition}

\subsection{Pairwise Closed Neighborhood Shortest Path Tree}
\label{pairwise}

Next we show the approximation performance of the structure applied in the algorithm design--\emph{Pairwise Closed Neighborhood Shortest Path Tree}. 

\begin{definition} \textbf{(Pairwise Closed Neighborhood Shortest Path Tree)}
Given a connected undirected graph $G = (V, E)$ and a terminal set $S \subset V$, construct another graph $G'$ whose vertex set is $S$. In $G',$ the edge weight is determined by the closed neighborhood shortest path between its two end vertices in the original graph. The pairwise closed neighborhood shortest path tree is the tree on $G'$ with the minimum total number of closed neighborhoods on its edges.
\end{definition}

\begin{lemma} \cite{goel2009approximability}
\label{mist-submodular_path}
There exists a randomized algorithm that finds an $O(n^{\frac{2}{3}})$-approximate solution to the $s{\text-}t$ shortest path problem with submodular edge costs in polynomial time. 
\end{lemma}

\begin{lemma} \cite{goel2009approximability}
\label{mist-submodular_tree}
There exists an algorithm that finds an $n$-approximate solution for the submodular minimum spanning problem in polynomial time.
\end{lemma}

\begin{lemma} \cite{yao1977probabilistic} (\textbf{Yao's Minimax Principle})
\label{mist-yao}
For any randomized algorithm, there exists a probability distribution on inputs to the algorithm, so that the expected cost of the randomized algorithm on its worst-case input is at least as large as the cost of the best deterministic algorithm on a random input from this distribution. 
\end{lemma}

\begin{theorem}
\label{c-spt}
Given a submodular edge-weighted graph $G=(V, E)$ and a terminal set $S\subseteq V,$ there exists an algorithm that finds an $|S|\cdot O(n^{2/3})$-approximate solution for the pairwise submodular shortest path tree problem in polynomial time. 
\end{theorem}

\begin{proof}
For a given submodular edge-weighted graph $G=(V, E)$ and terminal set $S\subseteq V,$ constuct another graph $G'$ in the following way: start with the terminal vertex set $S$ from $G$, for each pair of vertices within this set, determine the submodular shortest path between them in $G$. In $G'$, the vertices remain the same as those in the terminal vertex set of $G$, but the edge weights are assigned based on the submodular shortest paths of their corresponding vertices in $G$. 

In graph $G,$ for any pair of vertices $u, v$ in a given terminal set $S\subseteq V,$ according to Lemma \ref{mist-submodular_path}, there exists a randomized algorithm that finds an $O(n^{\frac{2}{3}})$-approximate solution to the $u{\text-}v$ shortest path problem with submodular edge costs. Based on Lemma \ref{mist-submodular_tree}, there exists an algorithm that finds an $|S|$-approximate solution for the submodular minimum spanning problem in polynomial time, for the given graph $G'$ with $|S|$ vertices, if the submodular edge weights are known. Based on the construction above, each edge weight of $G'$ is still submodular since the class of submodular functions is closed under non-negative linear combinations \cite{fujishige2005submodular}, and it is $O(n^{\frac{2}{3}})$-approximate. Therefore there exists an algorithm that finds an $|S|\cdot O(n^{2/3})$-approximate solution for the pairwise submodular shortest path tree problem. Since $|S|\leq n,$ $|S|\cdot O(n^{2/3})$ can also be represented as $O(n^{5/3}).$ 

Additionally, according to Yao's Minimax Principle, the aforementioned randomized algorithms can be derandomized without sacrificing their quality.
\end{proof}

In Lemma \ref{mist-edge-sub}, when the terminal set only contains two vertices $S=\{s, t\}$ the above Lemma still holds, therefore, it applies to any path too. We can get the following Corollary for the \emph{Pairwise Closed-Neighborhood Shortest Path Tree}. 

\begin{corollary}
Given an edge-weighted graph $G=(V, E)$ and a terminal set $S\subseteq V,$ there exists an algorithm that finds an $|S|\cdot O(n^{2/3})$-approximate solution for the Pairwise Closed-Neighborhood Shortest Path Tree problem. 
\end{corollary}

\subsection{Two-Stage Submodular Relaxation Algorithm}

To integrate the two objectives, let's start by defining the data structure of the algorithm. 

In a wireless mesh network represented by $G=(r, V, E, F, L)$, where the latency matrix is denoted as $\mathcal{M}_{n\times n}$ with $\mathcal{M}_{ij}=L(e_{ij})$ if $e_{ij}$ exists in $G$, and $\mathcal{M}_{ij}=\infty$ otherwise, $i,j=0,...,n-1$; and the edge closed neighborhood matrix with set elements is $\mathcal{N}_{n\times n}$ with $\mathcal{N}_{ij}=N[v_i]\cup N[v_j]$ if $e_{ij}$ exists in $G$, and $\mathcal{N}_{ij}=\emptyset$ otherwise, $i,j=0,...,n-1$.

In $G=(V, E)$ with the terminal set $S\subseteq V,$ the sets of Steiner Trees, Spanning Trees, and $s{\text-}t$ Shortest Paths exhibit an exponential size. However, this paper employs the \emph{Value Oracle} model, the same as in \cite{goel2009approximability}, in which, an oracle is utilized to provide the value of the submodular function $f(A)$ in response to queries involving the set $A\subseteq \mathcal{U},$ where $\mathcal{U}$ is the universe set.

In multicriteria optimization, managing multiple, potentially conflicting objectives can indeed be challenging. One natural approach to address this complexity is to prioritize one objective as the primary focus while treating the others as budgeted constraints. This allows for a more manageable optimization process where the primary objective is optimized subject to the constraints imposed by the secondary objectives. The following paper \cite{grandoni2014new} has discussions of this technique. 

In this reformulation, the problem MIST is transformed into a single-objective optimization problem by converting one of the objective functions into a constraint with a specified budget. The primary objective is to minimize the total number of closed neighborhoods in the Steiner tree. The Steiner tree is constructed with the terminal vertex set as the required vertices, comprising the union of the root vertex and the vertices hosting the required virtual functions. Additionally, the total latency length of the Steiner tree is constrained by a budget. 

With the properties of the MIST problem investigated above, a two-stage submodular relaxation algorithm has been designed. In the first stage, we compute the all pairs shortest paths (APSP) of vertices in $G$, then utilize the subroutine of \emph{Closed Neighborhood $i{\text-}j$ Shortest Path} for pairs of terminal vertices, where the terminals consist of the union of the root vertex and the vertices hosting the required virtual functions. This subroutine also ensures that the total latency length between the terminals does not exceed a specified constant times the length of the shortest path between them. In the second stage, we employ the \emph{Pairwise Closed Neighborhood Minimum Spanning Tree}, whose construction has been detailed in Section \ref{pairwise}.

\begin{algorithm}[H]
\caption*{\textbf{Two-Stage Submodular Relaxation Algorithm (TSSR)}}\label{alg:tssr}
\begin{algorithmic}[1]
\Require Latency matrix $\mathcal{M}_{n\times n}$ of Wireless Mesh Network $G=(r, V, E, F, L);$ multicast request $R=[\mathcal{F}_1, \mathcal{F}_2,...,\mathcal{F}_k].$
\Ensure A Steiner tree $T=(V_T, E_T)$ rooted at $r$, such that $r\cup V_R=S\subseteq V_T\subseteq V.$
\State $T=\emptyset$ 
\State $L(T)=0$
\State $N[T]=\emptyset$
\State Get the All Pairs Shortes Paths (APSP) matrix  on latency matrix $\mathcal{M}_{n\times n}$ of $G:$ $\mathcal{M}^{APSP}$
\For {$i = 1$ to $|r\cup V_R|$}
		\For {$j = i+1$ to $|r\cup V_R|$}
        \While{$L(P_{ij})\leq \alpha\mathcal{M}^{APSP}_{ij},~where~ \alpha\geq 1,$}
		\State Execute \emph{\textbf{Closed Neighborhood $i{\text-}j$ Shortest Path}} on $G,$ obtain $N[P_{ij}]$
        \EndWhile
        \State $N[P_{ij}^{*}]=\arg\min|N[P_{ij}]|$
        \State $L(P_{ij}^{*})=L(\arg\min|N[P_{ij}]|)$
		\EndFor
\EndFor
\State Construct the \emph{\textbf{Pairwise Closed Neighborhood Shortest Path Tree}} $T'$ of $r\cup V_R$ on $G$
\State $|N[T]|=|N[T]|+|N[T']|$
\State $L(T)=L(T)+L(T')$
\State 
\Return $L(T)$ and $|N[T]|$ 
\end{algorithmic}
\end{algorithm}

\section{Simulation}
\label{simulation}

\subsection{Simulation Setup}

In this section, we evaluate the performance of TSSR through simulations conducted on the topology and multicast requests, as outlined in Table \ref{tab:Table 1}. It conducts comparisons with Python NetworkX default heuristic approaches of Steiner Tree (ST) that do not consider interference. Assume the root is always $N_1,$ each edge has length $1$ and $\alpha=1.$ 

\begin{table}[H]
\caption{SDN Empowered WMN Topologies and User Equipment Multicast Requests }
\label{tab:Table 1}
\centering
\adjustimage{width=5.5cm,valign=c}{Ring-Star}
\begin{tabular}{|c|c|}
	\hline
	\textbf{$i$} & \textbf{$R_i$}\\ \hline
	1          & {[}$\mathcal{F}_1, \mathcal{F}_2${]}\\ \hline
	2          & {[}$\mathcal{F}_1, \mathcal{F}_3${]}\\ \hline
	3          & {[}$\mathcal{F}_1, \mathcal{F}_4${]}\\ \hline
	4          & {[}$\mathcal{F}_2, \mathcal{F}_3${]}\\ \hline
	5          & {[}$\mathcal{F}_2, \mathcal{F}_4${]}\\ \hline
	6          & {[}$\mathcal{F}_3, \mathcal{F}_4${]}\\ \hline
	7          & {[}$\mathcal{F}_1, \mathcal{F}_2, \mathcal{F}_3${]}\\ \hline
	8          & {[}$\mathcal{F}_1, \mathcal{F}_2, \mathcal{F}_4${]}\\ \hline
	9          & {[}$\mathcal{F}_1, \mathcal{F}_3, \mathcal{F}_4${]}\\ \hline
	10          & {[}$\mathcal{F}_2, \mathcal{F}_3, \mathcal{F}_4${]}\\ \hline
	11          & {[}$\mathcal{F}_1, \mathcal{F}_2, \mathcal{F}_3, \mathcal{F}_4${]}\\ \hline
\end{tabular}
\end{table}

\subsection{Simulation Results}

The simulation and comparison results are summarized in Table \ref{tab:Table 2}, where the length of the multicast tree and the interference for each request are presented. For instance, $R_1$ only involves $F_1$ and $F_2$ which are located at node $N_1$ and $N_2$ respectively, and there is an edge connecting the two nodes. As a result, the multicast tree found by the proposed approach is the same as that by ST. For the interference metric, node $N_1$ has three edges that can interfere with each other during transmission, so has $N_2$. This gives in total of 6 for the interference metric. The comparison of results in Table \ref{tab:Table 2} indicates that the proposed approach is able to construct a multicast tree of the same length, while at the same time producing reduced interference in comparison with ST approach.

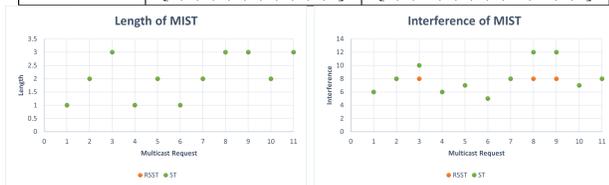
\begin{table}[H]
\centering
\caption{Simulation Results}
\begin{tabular}{|c|c|c|}
	\hline
	\textbf{Algorithm}  & \textbf{Length for $R_i$} & \textbf{Interference for $R_i$} \\ \hline
	TSSR   & [1,2,3,1,2,1,2,3,3,2,3] & [6,8,8,6,7,5,8,8,8,7,8]  \\ \hline
	ST     & [1,2,3,1,2,1,2,3,3,2,3] & [6,8,10,6,7,5,8,12,12,7,8]  \\ \hline
	\end{tabular}
 \adjustimage{width=4cm,valign=c}{Length}
	\adjustimage{width=4cm,valign=c}{Interference}
\label{tab:Table 2}
\end{table}

\section{Conclusion}
\label{con}
In this paper, we introduced an original approach for multicast in wireless mesh networks, which not only supports network functions implemented at different nodes for the multicast service but also considers interference. The bicriteria minimum interference Steiner tree (MIST) formulation, along with the demonstrated properties and algorithmic design techniques, contributes to a comprehensive understanding of the problem. These provide valuable insights for future studies in this research direction. In the paper, to ease the representation and simplify the analysis, we have focused on one multicast request at a time and assumed a fixed access point serving as the root of the Steiner tree. For future research, these assumptions may be relaxed, e.g., including geographical proximity rules and handling multiple multicast orchestrations.

\paragraph*{Acknowledgments} The authors thank Dr.~Benjamin Raichel for his insightful discussions regarding the non-metric properties of the two objectives.

\bibliographystyle{IEEEtran}
\bibliography{reference.bib}

\end{document}